\documentclass[english,nofootinbib]{revtex4-1}
\usepackage{lmodern}

\usepackage[T1]{fontenc}
\usepackage[utf8]{inputenc}
\usepackage{xcolor}
\usepackage{babel}
\usepackage{array}
\usepackage{verbatim}
\usepackage{enumitem}
\usepackage{multirow}
\usepackage{amsmath}
\usepackage{amsthm}
\usepackage{amssymb}
\usepackage{graphicx}
\usepackage[unicode=true,pdfusetitle,
 bookmarks=true,bookmarksnumbered=false,bookmarksopen=false,
 breaklinks=false,pdfborder={0 0 0},backref=false,colorlinks=true]
 {hyperref}
\hypersetup{
 pdfborderstyle={},linkcolor=blue,citecolor=blue,urlcolor=blue}

\makeatletter

\providecommand{\tabularnewline}{\\}

  \theoremstyle{plain}
  \newtheorem{assumption}{\protect\assumptionname}
  \theoremstyle{plain}
  \newtheorem{lem}{\protect\lemmaname}
\theoremstyle{plain}
\newtheorem{thm}{\protect\theoremname}
  \theoremstyle{plain}
  \newtheorem{rem}{\protect\remarkname}
  \theoremstyle{plain}
  \newtheorem{prop}{\protect\propositionname}

\@ifundefined{showcaptionsetup}{}{%
 \PassOptionsToPackage{caption=false}{subfig}}
\usepackage{subfig}
\makeatother

  \providecommand{\assumptionname}{Assumption}
  \providecommand{\lemmaname}{Lemma}
  \providecommand{\propositionname}{Proposition}
  \providecommand{\remarkname}{Remark}
\providecommand{\theoremname}{Theorem}

\begin{document}

\title{Influence of intrinsic spin in the formation of singularities for
inhomogeneous effective dust space-times}

\author{Paulo Luz}
\email{paulo.luz@ist.utl.pt}
\affiliation{Centro de Matemática, Universidade do Minho, Campus de Gualtar, 4710-057
Braga, Portugal}

\affiliation{Centro de Astrofísica e Gravitação - CENTRA, Departamento de Física,
Instituto Superior Técnico - IST, Universidade de Lisboa - UL, Av. Rovisco
Pais 1, 1049-001 Lisboa, Portugal}

\author{Filipe C. Mena}
\email{fmena@math.uminho.pt}
\affiliation{Centro de Matemática, Universidade do Minho, Campus de Gualtar, 4710-057
Braga, Portugal}

\affiliation{Departamento de Matem\'atica, Instituto Superior Técnico, Universidade de Lisboa, Av. Rovisco
Pais 1, 1049-001 Lisboa, Portugal}

\author{Amir Hadi Ziaie}
\email{ah.ziaie@riaam.ac.ir}
\affiliation{Research~Institute~for~Astronomy~and~Astrophysics~of  Maragha~(RIAAM),~P.O.~Box~55134-441,~Maragha,~Iran}
\begin{abstract}
The evolution of inhomogeneous space-times composed of uncharged fermions
is studied for Szekeres metrics which have no Killing vectors, in general. Using the Einstein-Cartan theory
to include the effects of (intrinsic) matter spin in General Relativity,
the dynamics of a perfect fluid with non-null spin degrees of freedom
is considered. It is shown that, if the matter is composed by effective
dust and certain constraints on the initial data are verified, a singularity
will not form. Various special cases are discussed, such as Lemaître-Tolman-Bondi
and Bianchi $I$ space-times, where the results are further extended
or shown explicitly to be verified.
\end{abstract}
\maketitle

\section{Introduction}

It is known that the process of gravitational collapse in General
Relativity (GR) can lead to spacetime singularities, which can correspond not only to black holes but also to naked
singularities, depending on the initial data. In turn, the latter can challenge the cosmic censorship conjecture which, in this context, has been tested for various matter contents such as scalar fields~\cite{sfcoll},
perfect fluids~\cite{pfcoll} and imperfect fluids~\cite{impflcoll}.

It would then be interesting to explore alternative theories of gravity
and test whether the results obtained within GR are robust with respect
to small deviations from the theory. In particular, it is important to verify
if corrections to GR may prevent the formation of singularities. Work
along these lines has been done in alternative theories of gravity such
as $F(R)$ gravity~\cite{frcoll}, Brans-Dicke theory~\cite{BDcoll},
Lovelock~\cite{Lovecoll} and Gauss-Bonnet gravity~\cite{GBcoll},
where it was found that the geometrical attributes absent in
GR could affect the final fate of collapse.

In this paper, we consider the formalism provided by the Einstein-Cartan
theory (ECT), where it is possible to introduce (intrinsic) spin degrees
of freedom in a geometric theory of gravity. Such model is specially
interesting because it encapsulates such quantum corrections in a
semi-classical limit \cite{Hehl1,Hehl2}, providing a way to infer
if quantum effects may avoid the formation of singularities.

Indeed, in the framework of loop quantum gravity, it was found that
for a universe permeated by a scalar field, quantum corrections modify
the classical Friedmann equation and a cosmological singularity can
be avoided \cite{ashsingh,quantumbounce}. Moreover, some interesting
but somewhat forgotten results from the 1970s \cite{Trautman,Stewart,Kop}
suggest that the introduction of spin in spatially homogeneous anisotropic
models may prevent the formation of singularities, depending on the
degree of isotropy of the collapsing space-time. However, in ref. \cite{Hehl1}, by applying the strong energy conditions to an effective energy-momentum tensor,
 it was shown that modified versions of the Hawking-Penrose singularity theorems hold, even considering the effects of spin.

%However, by modifying the Hawking-Penrose singularities theorems, such that the strong energy conditions are applied to an effective energy-momentum tensor, it was shown \cite{Hehl1} that singularities are unavoidable, even considering the effects of spin (provided the space-time is globally hyperbolic and a trapped surface exists).

More recently, the collapse process of a spatially homogeneous and isotropic spin fluid has
been studied in \cite{astrosptor}, where it was shown that there exists a competition between
fluid and spin source parameters, so that the collapse end state is
determined via the dynamics of these source terms and both
singular and nonsingular solutions could be obtained. Moreover, the
effects of spin have been shown to come into play in the cosmological
context, where the torsion of space-time generated by spinning particles
induces gravitational repulsion in fermionic matter in the early universe,
where extremely high densities are present avoiding then the formation
of space-time singularities \cite{Magueijo}.

Nonetheless, spatially homogeneous and isotropic models may be seen as simplistic,
so it would be interesting to explore this debate further by considering
inhomogeneous and anisotropic space-times. Towards this goal, in this paper,
we take a natural first step by considering effective inhomogeneous dust equations
of state. This is surely a critical case which, yet, should
be explored by comparison with the GR case, known to develop naked
singularities.

We show that, with respect to the previous results, the situation
can be more subtle: the effects of spin contribute, in a way, as a
repulsive potential, therefore, collapsing solutions might not even
exist for inhomegeneous space-times. Furthermore, even if some form of the energy conditions is
verified, singularities may not be the end result of gravitational
collapse if spin effects are considered.

The article is organized as follows: in Sec.~\ref{sec:Inclusion_of_spin}
we include the effects of spin in the ECT and show that, in the considered
setup, the resultant model is equivalent to an effective perfect fluid
in the theory of GR; in Sec.~\ref{sec:Spin_effects_in_the_evolution_of_the_collapse}
we consider the case of an effective dust fluid in a Szekeres space-time
and study how the (intrinsic) spin degrees of freedom affect the evolution
of the space-time and the formation of singularities; in Sec.~\ref{sec:Special_Cases}
we consider some particular solutions of the Szekeres model where
previous results can be extended or shown to be verified explicitly;
in Sec.~\ref{sec:Conclusion} we summarize the results and conclude.

In this article all quantities are expressed in the geometrized unit
system where $\mathcal{G}=c=1$, greek indices range from $0$ to
$3$ and we use the metric signature $\left(-+++\right)$.

\section{Inclusion of spin effects in General Relativity\label{sec:Inclusion_of_spin}}

\subsection{The Einstein-Cartan theory in brief}

To include the spin effects in GR we shall follow the semi-classical
approach provided by ECT and relate the spin of the matter fields
with the space-time torsion tensor field. Now, let us briefly review
the ECT and see how spin effects can be included in GR (for more details
see e.g. \cite{Hehl1,Hehl2}).

We start by recalling that, in this context, the covariant derivative
of a vector field $U$ is given by 
\begin{equation}
\nabla_{\alpha}U^{\beta}=\partial_{\alpha}U^{\beta}+C_{\alpha\sigma}^{\beta}U^{\sigma}\,,\label{eq:Covariant_definition}
\end{equation}
constrained to be metric compatible, $\nabla_{\alpha}g_{\beta\gamma}=0$,
but otherwise with a generic connection $C_{\alpha\beta}^{\gamma}$.
The anti-symmetric part of the connection defines a tensor field,
called the \emph{torsion tensor field}, 
\begin{equation}
S_{\alpha\beta}\,^{\gamma}\equiv C_{\left[\alpha\beta\right]}^{\gamma}=\frac{1}{2}\left(C_{\alpha\beta}^{\gamma}-C_{\beta\alpha}^{\gamma}\right)\,.\label{eq:Torsion_tensor}
\end{equation}
From this definition, it is possible to write the connection $C_{\alpha\beta}^{\gamma}$
as a combination of the torsion tensor plus the usual metric Christoffel
symbols $\Gamma_{\alpha\beta}^{\gamma}$, such that $C_{\alpha\beta}^{\gamma}=\Gamma_{\alpha\beta}^{\gamma}+S_{\alpha\beta}\,^{\gamma}+S^{\gamma}\,_{\alpha\beta}-S_{\beta}\,^{\gamma}\,_{\alpha}$,
or 
\begin{equation}
C_{\alpha\beta}^{\gamma}=\Gamma_{\alpha\beta}^{\gamma}+K_{\alpha\beta}{}^{\gamma}\,,\label{eq:Connection_general}
\end{equation}
where the tensor field, 
\begin{equation}
K_{\alpha\beta}\,^{\gamma}\equiv S_{\alpha\beta}{}^{\gamma}+S^{\gamma}{}_{\alpha\beta}-S_{\beta}{}^{\gamma}{}_{\alpha}\,,\label{eq:Contorsion_general}
\end{equation}
is called \emph{contorsion tensor}.

Now, from the definition of the Riemann tensor 
\begin{equation}
R_{\alpha\beta\gamma}{}^{\rho}=\partial_{\beta}C_{\alpha\gamma}^{\rho}-\partial_{\alpha}C_{\beta\gamma}^{\rho}+C_{\beta\sigma}^{\rho}C_{\alpha\gamma}^{\sigma}-C_{\alpha\sigma}^{\rho}C_{\beta\gamma}^{\sigma}\,,
\end{equation}
and Eq.~\eqref{eq:Connection_general}, we can write the Riemann
tensor in ECT in terms of the metric Riemann tensor, $\widetilde{R}_{\alpha\beta\gamma}{}^{\rho}$,
defined just in terms of the Christoffel symbols, and the contorsion
tensor as 
\begin{equation}
R_{\alpha\beta\gamma}{}^{\rho}=\widetilde{R}_{\alpha\beta\gamma}{}^{\rho}+2\widetilde{\nabla}_{\left[\beta\right.}K_{\left.\alpha\right]\gamma}{}^{\rho}+2K_{\left[\beta\right|\sigma}{}^{\rho}K_{\left|\alpha\right]\gamma}{}^{\sigma}\,,\label{eq:Riemann_tensor}
\end{equation}
where the square brackets denote anti-symmetrization in the indicated
indexes and $\widetilde{\nabla}_{\alpha}$ represents the covariant
derivative of a tensor quantity using only the metric connection $\Gamma_{\alpha\beta}^{\gamma}$.
Contracting Eq.~\eqref{eq:Riemann_tensor} gives the following expression
for the Ricci tensor 
\begin{equation}
R_{\alpha\gamma}=\widetilde{R}_{\alpha\gamma}+2\widetilde{\nabla}_{\left[\beta\right.}K_{\left.\alpha\right]\gamma}{}^{\beta}+2K_{\left[\beta\right|\sigma}{}^{\beta}K_{\left|\alpha\right]\gamma}{}^{\sigma}\,,\label{eq:Ricci_tensor}
\end{equation}
whereby the action for ECT which is linear with respect to curvature
scalar can be written in the form 
\begin{equation}
A=\frac{1}{16\pi}\int\sqrt{-g}\,g^{\alpha\gamma}\left\{ \widetilde{R}_{\alpha\gamma}+K_{\beta\sigma}{}^{\beta}K_{\alpha\gamma}{}^{\sigma}-K_{\alpha\sigma}{}^{\beta}K_{\beta\gamma}{}^{\sigma}\right\} \,d^{4}x+\int\sqrt{-g}\,\mathcal{L}_{\text{Matter}}\,d^{4}x\,,\label{eq:Action}
\end{equation}
where $\mathcal{L}_{\text{Matter}}$ represents the matter Lagrangian
density and we have omitted the total derivative terms. Following
the Palatini approach, where the action is varied independently with
respect to the space-time metric and the connection, we find two sets
of field equations. Making the variation of the action in Eq.~\eqref{eq:Action}
with respect to the contorsion tensor we find 
\begin{equation}
S^{\alpha\beta\mu}+2g^{\mu[\alpha}S^{\beta]}=-8\pi\Delta^{\alpha\beta\mu}\,,\label{eq:torsion_FE_general}
\end{equation}
where $S_{\alpha}\equiv S_{\mu\alpha}{}^{\alpha}$ and the quantity
\begin{equation}
\Delta^{\alpha\beta\mu}=\frac{1}{\sqrt{-g}}\frac{\delta\left(\sqrt{-g}\mathcal{L}_{\text{Matter}}\right)}{\delta K_{\mu\beta\alpha}}\,,\label{eq:hypermomentum_general}
\end{equation}
is usually called the \emph{intrinsic hypermomentum}, as it encapsulates
all the information of the microscopic structure of the matter that
composes the fluid, i.e. intrinsic spin, dilaton charge and intrinsic
shear. Notice that, in Eq.~\eqref{eq:torsion_FE_general}, the torsion
tensor field is not a dynamical field, in the sense that the left
hand side contains no derivatives of the torsion tensor and indeed
appears as a purely algebraic equation. As such, in vacuum, where
the matter Lagrangian density $\mathcal{L}_{\text{Matter}}$ is null,
from Eq.~\eqref{eq:torsion_FE_general}, we see that the torsion
tensor is effectively null so, all solutions of GR in vacuum are also
solutions of ECT \cite{Hehl2}.

On the other hand, making the variation of the action in Eq.~\eqref{eq:Action}
with respect to the metric tensor we find the modified Einstein field
equations 
\begin{equation}
\begin{aligned}\widetilde{G}_{\mu\nu} & -4S_{\sigma}K_{\mu\nu}{}^{\sigma}-2g^{\beta\delta}K_{\mu\sigma\delta}K_{\beta\nu}{}^{\sigma}+2g_{\mu\nu}S_{\sigma}S^{\sigma}+\\
 & +\frac{1}{2}g_{\mu\nu}K_{\rho\sigma\delta}K^{\delta\rho\sigma}-4S_{\mu}S_{\nu}-g^{\sigma\rho}K_{\beta\rho\nu}K_{\sigma\mu}{}^{\beta}=8\pi\,T_{\mu\nu}\,,
\end{aligned}
\label{eq:EFE_general_torsion}
\end{equation}
where $\widetilde{G}_{\mu\nu}\equiv\widetilde{R}_{\mu\nu}-\frac{1}{2}g_{\mu\nu}\widetilde{R}$
represents the Einstein tensor defined just using the metric connection,
and 
\begin{equation}
T_{\mu\nu}=g_{\mu\nu}\mathcal{L_{\text{Matter}}}-2\frac{\delta\mathcal{L}_{\text{Matter}}}{\delta g^{\mu\nu}}\,,\label{eq:energy-momentum_tensor_general}
\end{equation}
is the (symmetric) energy-momentum tensor of the matter fluid.

\subsection{The Weyssenhoff fluid}

To relate the torsion tensor field with the intrinsic spin of matter
particles, we shall consider a perfect fluid with non-null (intrinsic)
spin degrees of freedom that permeates a region of space-time. A description
of such fluid is given by the Weyssenhoff model \cite{Weyssenhoff}.

Although the Weyssenhoff fluid was first proposed as a phenomenological
theory, a variational theory of an ideal spinning fluid has been developed
\cite{Smalley,obuk} providing a Lagrangian for the Weyssenhoff semi-classical
spin fluid which can, in turn, be used to compute the corresponding
hypermomentum and energy-momentum tensors through Eqs.~\eqref{eq:hypermomentum_general}
and \eqref{eq:energy-momentum_tensor_general}. It is then found that
the Weyssenhoff fluid is characterized by the following relation between
the intrinsic hypermomentum tensor and spin \cite{Smalley,obuk}
\begin{equation}
\Delta^{\mu\nu\alpha}=\frac{1}{2}\tau^{\mu\nu}u^{\alpha}\,,\label{eq:hypermomentum_spin_relation}
\end{equation}
where $u^{\alpha}$ is the fluid's 4-velocity and $\tau^{\mu\nu}$
is an anti-symmetric tensor representing the spin density tensor which,
in turn, verifies the following constraint 
\begin{equation}
\tau^{\alpha\beta}u_{\beta}=0\,,\label{eq:Frenkel_condition}
\end{equation}
known as the \emph{Frenkel condition} \cite{Frenkel1,Frenkel2}. Substituting
Eqs.~\eqref{eq:hypermomentum_spin_relation} and \eqref{eq:Frenkel_condition}
in Eq.~\eqref{eq:torsion_FE_general}, we find that the torsion tensor
is related to the intrinsic spin of matter as 
\begin{equation}
S^{\alpha\beta\mu}=-4\pi\,\tau^{\alpha\beta}u^{\mu}\,,\label{eq:spin_torsion_relation}
\end{equation}
making it clear that, in this model, spin constitutes the effective
source of the torsion field. Eqs.~\eqref{eq:Frenkel_condition} and
\eqref{eq:spin_torsion_relation} allow us then to rewrite Eq.~\eqref{eq:EFE_general_torsion}
as 
\begin{equation}
\begin{aligned}\widetilde{G}_{\mu\nu} & +16\pi^{2}\tau_{\alpha\beta}\tau^{\alpha\beta}u_{\mu}u_{\nu}+8\pi^{2}g_{\mu\nu}\tau_{\alpha\beta}\tau^{\alpha\beta}\\
 & -32\pi^{2}\tau_{\mu\sigma}\tau_{\nu}{}^{\sigma}=8\pi\,T_{\mu\nu}\,.
\end{aligned}
\label{eq:EFE_intermediate1}
\end{equation}
Now, the spin tensor $\tau^{\mu\nu}$ and the energy-momentum tensor
$T_{\mu\nu}$ that appear in Eqs.~\eqref{eq:hypermomentum_spin_relation}
- \eqref{eq:EFE_intermediate1} refer to the spin and energy-momentum
of microscopic particles. However, we are interested in studying the
macroscopic behavior of an ideal spinning fluid. Therefore, to find
the field equations that describe a macroscopic gravitational field
due to a Weyssenhoff fluid, we have to compute a space-time average
of the tensors $\tau^{\mu\nu}$ and $T_{\mu\nu}$ over an element
of volume of the fluid, respectively $\langle\tau^{\mu\nu}\rangle$
and $\langle T_{\mu\nu}\rangle$.

Assuming that the spinning fluid is composed of microscopic particles
with randomly oriented spin, we get the averaged quantities \cite{Gasperini}
\begin{align}
\langle\tau^{\mu\nu}\rangle & =0\,,\label{eq:average_spin_tensor1}\\
\langle\tau^{\alpha\beta}\tau_{\alpha\beta}\rangle & =2s^{2}\,,\\
\langle\tau_{\mu}{}^{\alpha}\tau_{\nu\alpha}\rangle & =\frac{2}{3}\left(g_{\mu\nu}+u_{\mu}u_{\nu}\right)s^{2}\,,
\end{align}
where $s^{2}$ represents the average of the square of the spin density
of the fluid. Moreover, from the Lagrangian density given in \cite{obuk} for a zero-vorticity fluid, we find 
\begin{equation}
\left\langle T_{\mu\nu}\right\rangle =-\frac{8\pi}{3}u_{\mu}u_{\nu}s^{2}-\frac{8\pi}{3}g_{\mu\nu}s^{2}+\left(\rho+p\right)u_{\mu}u_{\nu}+p\,g_{\mu\nu}\,,\label{eq:average_energy-momentum_tensor}
\end{equation}
where $\rho$ and $p$ represent the mass-energy density and pressure
of the fluid, respectively. Substituting Eqs.~\eqref{eq:average_spin_tensor1}
- \eqref{eq:average_energy-momentum_tensor} in the average version
of Eq.~\eqref{eq:EFE_intermediate1} we find 
\begin{equation}
\widetilde{G}_{\mu\nu}=8\pi\,\left[\left(\rho+p-4\pi s^{2}\right)u_{\mu}u_{\nu}+\left(p-2\pi s^{2}\right)\,g_{\mu\nu}\right]\,.\label{eq:EFE_intermediate2}
\end{equation}
Notice that the field equations in Eq.~\eqref{eq:EFE_intermediate2},
for a zero-vorticity Weyssenhoff fluid in the Einstein-Cartan theory,
are actually equivalent to those in GR for a perfect fluid with additional
contributions from the spin to the energy-density and pressure. Therefore,
a spinning perfect fluid can be, classically, described by the theory
of GR assuming that the fluid is described by the effective energy-momentum
tensor 
\begin{equation}
T_{\mu\nu}=\left(\rho_{\text{eff}}+p_{\text{eff}}\right)u_{\mu}u_{\nu}+p_{\text{eff}}\,g_{\mu\nu}\,,\label{eq:Energy_momentum_tensor_fluid}
\end{equation}
with 
\begin{equation}
\rho_{{\rm eff}}=\rho-2\pi s^{2}\qquad\text{and}\qquad p_{{\rm eff}}=p-2\pi s^{2},\label{eq:eff_density}
\end{equation}
where we have omitted the angular brackets. Nevertheless, one should
bear in mind that all quantities are to be considered as averages.

Moreover, from now onwards, we will also omit the tilde to indicate
tensor quantities computed only using the symmetric part of the connection
since, as was shown, we can effectively treat the problem at hand
using the theory of GR, where the torsion tensor is null.

\subsection{The energy-momentum tensor for effective uncharged spinning dust}

We shall now restrict our attention to the case where the fluid's
pressure varies in such a way that it cancels the contribution coming
from spin effects, so that the whole fluid effectively behaves as
dust, i.e. 
\begin{equation}
p_{\text{eff}}=p-2\pi s^{2}=0\,.\label{eq:eff_pressure_null}
\end{equation}
so that 
\begin{equation}
T_{\text{eff}}^{\mu\nu}=\rho_{\text{eff}}u^{\mu}u^{\nu}\,,\label{eq:Energy_momentum_tensor_dust}
\end{equation}
where $\rho_{\text{eff}}$ represents the effective mass-energy density
of the matter.

Now, let us consider that the matter is composed of $N$ species of
uncharged fermionic particles with mass $m_{i}$ and spin $s_{i}=\pm\hslash/2$
and assume that the interactions between the microscopic constituents
of the fluid are negligible. Clearly, the distribution of spin and
mass are related to each other. The particle number density for each
species is given by \cite{Hehl1} 
\begin{equation}
n_{i}=\frac{\rho_{i}}{m_{i}}=\left(\frac{4s_{i}^{2}}{\hslash^{2}}\right)^{\frac{1}{2}}\,,\label{eq:number_density_spin_relation}
\end{equation}
where $\rho_{i}$ and $s_{i}^{2}$ represent, respectively, the averaged
energy density and the average of the squared spin density of a single
species, such that 
\begin{equation}
\rho_{i\text{eff}}=\rho_{i}-2\pi s_{i}^{2}\,,\qquad\text{and}\qquad\rho_{\text{eff}}=\sum_{i=1}^{N}\rho_{i\text{eff}}\,.\label{eq:eff_density_each}
\end{equation}
Using Eq.~\eqref{eq:number_density_spin_relation} in Eqs.~\eqref{eq:eff_density} and \eqref{eq:eff_density_each}
, we find that for each species 
\begin{equation}
\rho_{i\text{eff}}=\rho_{i}\left(1-\frac{\rho_{i}}{\bar{\rho}_{i}}\right)\,,\qquad p_{i\text{eff}}=p_{i}-\frac{\rho_{i}^{2}}{\bar{\rho}_{i}}\,,\label{eq:eff_density_final_polytrope}
\end{equation}
where the total pressure of the perfect fluid is formally given by
\begin{equation}
p_{\text{eff}}=\sum_{i=1}^{N}p_{i\text{eff}}\,,\label{eq:total_pressure_sum}
\end{equation}
and 
\begin{equation}
\bar{\rho}_{i}=\frac{2m_{i}^{2}}{\pi\hslash^{2}}\,\label{eq:Critical_density_general}
\end{equation}
represents a critical mass-density which sets a scale at which the
spin effects have to be taken into account \cite{Hehl1}.

In the perfect fluid approximation, the various constituents of the
matter are non-interacting. Therefore, the total pressure of the fluid
will be null if the various terms in Eq.~\eqref{eq:total_pressure_sum}
are null. Substituting Eq.~\eqref{eq:eff_density_final_polytrope}
in Eq.~\eqref{eq:eff_pressure_null}, we find that the constituents
of the fluid are characterized by an equation of state of the form
\begin{equation}
p_{i}=\frac{1}{\bar{\rho}_{i}}\rho_{i}^{2}\,,\label{eq:polytrope}
\end{equation}
that is, a polytrope of order 2.

\section{Spin effects in the gravitational evolution\label{sec:Spin_effects_in_the_evolution_of_the_collapse}}

To study the influence of spin in singularity formation, we analyze
the evolution of an uncharged effective dust fluid with non-null spin
degrees of freedom composed only of fermionic particles.

When considering the case of the gravitational collapse of a compact
object in astrophysics, we consider a model of Oppenheimer-Snyder
type having a collapsing interior given by the Szekeres metric \cite{Szekeres1,Szekeres2},
which is inhomogeneous, matched to a vacuum exterior. Examples of
such models were shown to exist in \cite{I_Brito,Tod-Mena,Bonnor}.
Whereas, when considering a cosmological model, we shall assume that
the coordinates in the Szekeres metric are globally defined and that
the universe can either initially be expanding or contracting. In
this case, part of the problem will be to figure out if, during evolution,
there can be a bounce in the collapsing phase or a turning point,
followed by recollapse, in the expanding phase.

The Szekeres space-time represents the solutions of the Einstein field equations (EFE) for a
space-time permeated by irrotational dust whose line element can be
written in the form 
\begin{equation}
ds^{2}=-d\tau^{2}+e^{\alpha\left(\tau,r,p,q\right)}dr^{2}+e^{\beta\left(\tau,r,p,q\right)}\left(dp^{2}+dq^{2}\right)\,,\label{eq:Szekeres_general_line_element}
\end{equation}
where $\alpha$ and $\beta$ are $C^{2}$-functions of the coordinates
$\tau,r,p,q$. This metric has no Killing vectors, in general \cite{Bonnor-2}. Historically, the Szekeres models are split in two
classes: one that generalizes the Lemaître-Tolman-Bondi (LTB) metrics
and another that generalizes the Kantowski-Sachs (KS) metrics. We
shall treat them separately.

\subsection{Szekeres space-times: LTB-like}

The LTB-like Szekeres models are characterized by a line element of
the form \cite{Szekeres1,Szekeres2} 
\begin{equation}
ds^{2}=-d\tau^{2}+\frac{\left(R'-\frac{RE'}{E}\right)^{2}}{\epsilon+f\left(r\right)}dr^{2}+\frac{R^{2}}{E^{2}}\left(dp^{2}+dq^{2}\right)\,,\label{eq:Szekeres_I_line_element}
\end{equation}
where the prime indicates a derivative with respect to $r$ , $\epsilon=\left\{ -1,0,1\right\} $,
$E=E\left(r,p,q\right)$ and $f\left(r\right)$ are arbitrary $C^{2}$
functions such that $f\left(r\right)+\epsilon>0$, while the function
$R=R\left(\tau,r\right)$ satisfies the Friedmann like equation

\begin{equation}
\dot{R}^{2}\left(\tau,r\right)=f\left(r\right)+\frac{2M\left(r\right)}{R\left(\tau,r\right)}\,,\label{eq:Szekeres_I_Friedmann}
\end{equation}
with the overdot denoting a derivative with respect to the time coordinate
$\tau$ and $M=M\left(r\right)$ is another arbitrary function.

Given the line element in Eq.~\eqref{eq:Szekeres_I_line_element}
and Eq.~\eqref{eq:Energy_momentum_tensor_dust} we find, from the EFE, the following relation between the
mass-energy density of the effective dust source and the functions
that describe the geometry of the space-time 
\begin{equation}
8\pi\,\rho_{\text{eff}}=2\frac{EM'-3ME'}{R^{2}\left(ER'-RE'\right)}\,.\label{eq:Szekeres_I_eff_density}
\end{equation}
This expression can be re-written as 
\begin{align}
\rho_{\text{eff}}\left(\tau,r,p,q\right) & =\frac{r^{2}\,F\left(r,p,q\right)}{R^{2}\left(\tau,r\right)W\left(\tau,r,p,q\right)}\,,\label{eq:Szekeres_I_eff_density_alt}
\end{align}
where 
\begin{equation}
r^{2}F\left(r,p,q\right)=\frac{1}{4\pi}\left[E\left(r,p,q\right)M'\left(r\right)-3M\left(r\right)E'\left(r,p,q\right)\right]\,,\label{eq:Szekeres_I_F_definition}
\end{equation}
and 
\begin{equation}
W\left(\tau,r,p,q\right)=E\left(r,p,q\right)R'\left(\tau,r\right)-R\left(\tau,r\right)E'\left(r,p,q\right)\,,\label{eq:Szekeres_I_function_W}
\end{equation}
so, $f\left(r\right)$, $M\left(r\right)$ and $E\left(r,p,q\right)$
will be free initial functions that characterize the space-time. In
turn, $E\left(r,p,q\right)$ is written in terms of another 3 free
functions (see \cite{Szekeres1,Szekeres2}), but we will not have
to use this fact in the what follows.

\subsubsection{Regularity conditions and the influence of spin in singularity formation\label{subsec:Szekeres_class_I_results}}

Since we are interested in studying the influence of spin in the formation
of curvature singularities, in addition to the previous assumptions
we shall consider some further requirements on the initial regularity
of the space-time:
\begin{assumption}
At the initial time $\tau=\tau_{0}$: \label{assu:assumption_1}\leavevmode
\begin{enumerate}
\item $R\left(\tau_{0},r\right)$ is an increasing function of the coordinate
$r$. 
\item The space-time is non-singular. 
\item For every triplet $(r,p,q)$, we assume $W\left(\tau_{0},r,p,q\right)$
and $\rho_{\text{eff}}\left(\tau_{0},r,p,q\right)$ to be non-null.
\end{enumerate}
\end{assumption}
Under these assumptions, we can set by convention 
\begin{equation}
R\left(\tau_{0},r\right)=r.\,\label{eq:coord_system_specified}
\end{equation}
When it exists, it is also useful to introduce $\tau_{c}\left(r\right)>\tau_{0}$,
defined as the value of the time coordinate at which the function
$R$ of the shell with coordinate $r$ goes to zero.

Under our assumptions, from Eqs.~\eqref{eq:Szekeres_I_eff_density_alt},
and in the coordinate system defined by Eq.~\eqref{eq:coord_system_specified},
it is straightforward to see that the function $F\left(r,p,q\right)$
must be finite and non-null. This leads us to conclude that a necessary
condition for the divergence of $\rho_{\text{eff}}\left(\tau,r,p,q\right)$
for a given $\left(r,p,q\right)$ is that either $R\left(\tau,r\right)$
or $W\left(\tau,r,p,q\right)$ go to zero at some instant of time $\tau_S>\tau_{0}$. The former is associated with a shell-focusing singularity,
whereas the latter represents the formation of a shell-crossing singularity
\cite{Szekeres}.

Moreover, since the functions $R\left(\tau,r\right)$ and $W\left(\tau,r,p,q\right)$
are continuous in the time coordinate $\tau$ before the singularity
formation, for each triad $\left(r,p,q\right)$, the effective mass-energy
density function $\rho_{\text{eff}}\left(\tau,r,p,q\right)$ is also
a continuous function in the coordinate $\tau$. We are then able
to prove 
\begin{lem}
\label{Lemma:Lemma_1} Under our assumptions if, in a continuous gravitational
collapse, there exists a curve $\tau=\tau_{s}(r)>\tau_{0}$ such that
for each $(r,p,q)$ within the matter, either $\lim_{\tau\to\tau_{s}(r)^{-}}R\left(\tau,r\right)=0$
or $\lim_{\tau\to\tau_{s}(r)^{-}}W\left(\tau,r,p,q\right)=0$ then,
$\lim_{\tau\to\tau_{s}(r)^{-}}\rho_{\text{eff}}\left(\tau,r,p,q\right)=-\infty$. 
\end{lem}

\begin{proof}
We split the proof in two parts:

1. First, we prove that for each triad $\left(r,p,q\right)$, if either
$\lim_{\tau\to\tau_{s}(r)^{-}}R\left(\tau,r\right)\to0$ or $\lim_{\tau\to\tau_{s}(r)^{-}}W\left(\tau,r,p,q\right)\to0$
for some instant $\tau_{s}>\tau_{0}$, then we have that $\lim_{\tau\to\tau_{s}(r)^{-}}\rho_{\text{eff}}\left(\tau,r,p,q\right)\to\pm\infty$.
As was discussed previously, a necessary condition for divergent $\rho_{\text{eff}}$
is that either $W\left(\tau,r,p,q\right)$ or $R\left(\tau,r\right)$
go to zero at some instant $\tau=\tau_{s}$. In appendix \ref{sec:Appendix-A}
it is shown that if $R$ is a real function, these conditions are
also sufficient for the divergence of $\rho_{\text{eff}}$, in particular,
even in the case where $\lim_{\tau\to\tau_{s}(r)^{-}}W\to\infty$,
the limit $\lim_{\tau\to\tau_{s}(r)^{-}}R^{2}R'\to0$ is always verified,
hence from \eqref{eq:Szekeres_I_function_W} $\lim_{\tau\to\tau_{s}(r)^{-}}R^{2}W\to0$
and from \eqref{eq:Szekeres_I_eff_density_alt} $\lim_{\tau\to\tau_{s}(r)^{-}}\rho_{\text{eff}}\left(\tau,r,p,q\right)\to\pm\infty$,
as $F(r,p,q)$ is finite. Therefore, since $F\left(r,p,q\right)$
takes finite values, from Eq.~\eqref{eq:Szekeres_I_eff_density_alt},
fixing $(r,p,q)$, if either $\lim_{\tau\to\tau_{s}(r)^{-}}R\left(\tau,r\right)\to0$
or $\lim_{\tau\to\tau_{s}(r)^{-}}W\left(\tau,r,p,q\right)\to0$ for
some instant $\tau_{s}>\tau_{0}$, then we have $\lim_{\tau\to\tau_{s}(r)^{-}}\rho_{\text{eff}}\left(\tau,r,p,q\right)\to\pm\infty$.

2. We now show that we can only have $\rho_{\text{eff}}\left(\tau,r,p,q\right)\to-\infty$.
For a dust spacetime composed only of fermions, the effective mass-energy
density is given by the sum in Eq.~\eqref{eq:eff_density_each} and
it takes infinite values if, at least, one term of the sum, $\rho_{i\text{eff}}$,
is infinite. Although we shall restrain ourselves from imposing any
of the usual energy conditions, we shall consider that if any of the
parameters that characterize the fluid take complex values at any
point of the space-time, the solution is unphysical. Therefore, for
$\rho_{i}\in\mathbb{R}\setminus\{0\}$, from Eq.~\eqref{eq:eff_density_final_polytrope},
we find that each $\rho_{i\text{eff}}$ may only tend to $-\infty$,
hence, for each $(r,p,q)$, $\lim_{\tau\to\tau_{s}(r)^{-}}\rho_{\text{eff}}\left(\tau,r,p,q\right)\to-\infty$. 
\end{proof}
Now, from Lemma \ref{Lemma:Lemma_1}, we get the following result:
\begin{thm}
\label{prop:Proposition_1}Given a Szekeres space-time with line element
\eqref{eq:Szekeres_I_line_element} permeated by an uncharged perfect
fluid composed only of fermionic particles, characterized by an equation of state
such that, the fluid effectively behaves as dust, if Assumption \ref{assu:assumption_1}
is verified and $\text{sign}\left(F\left(r,p,q\right)\right)=\text{sign}\left(W\left(\tau_{0},r,p,q\right)\right)$
then, a curvature singularity will not form.
\end{thm}
\begin{proof}
The proof follows from Lemma \ref{Lemma:Lemma_1} and the continuity
of $R$, $W$ and $\rho_{eff}$.

Consider that $\text{sign}\left(F\left(r,p,q\right)\right)=\text{sign}\left(W\left(\tau_{0},r,p,q\right)\right)$.

Let us start by showing that the function $W\left(\tau,r,p,q\right)$
will not be zero for any $\tau\in\left]\tau_{0},\tau_{c}\left(r\right)\right]$.
We argue by contradiction.

Assume that there exists, for each $r$, an instant $\tau_{1}\left(r\right)\in\left]\tau_{0},\tau_{c}\left(r\right)\right]$
such that $\lim_{\tau\to\tau_{1}^{-}}W\left(\tau_{1},r,p,q\right)=0$.
Then, from Lemma \ref{Lemma:Lemma_1}, $\lim_{\tau\to\tau_{1}^{-}}\rho_{\text{eff}}\left(\tau,r,p,q\right)\to-\infty$.
From our assumptions, $W\left(\tau,r,p,q\right)$ and $R\left(\tau,r\right)$
are continuous functions in the $\tau$ coordinate for all $\tau\in\left]\tau_{0},\tau_{1}\right[$
therefore, the effective mass-energy density, $\rho_{\text{eff}}$,
is a continuous function in the coordinate $\tau$ for all $\tau_{2}<\tau_{1}$,
concluding that $\rho_{\text{eff}}\left(\tau_{2},r,p,q\right)<0$.
On the other hand, continuity of $W\left(\tau,r,p,q\right)$ implies
that the sign of $W\left(\tau_{2},r,p,q\right)$ must be equal to
the sign of $W\left(\tau_{0},r,p,q\right)$. Therefore, since $R\left(\tau_{2},r\right)>0$
and $\text{sign}\left(F\left(r,p,q\right)\right)=\text{sign}\left(W\left(\tau_{2},r,p,q\right)\right)$,
we conclude, from Eq.~\eqref{eq:Szekeres_I_eff_density_alt}, $\rho_{\text{eff}}\left(\tau_{2},r,p,q\right)>0$,
contradicting what was shown before. Hence, $W\left(\tau,r,p,q\right)$
will not be zero for any $\tau\in\left]\tau_{0},\tau_{c}\left(r\right)\right]$.

The case when the function $R\left(\tau,r\right)$ goes to zero before
the function $W\left(\tau,r,p,q\right)$ also can not occur since,
from Eq.~\eqref{eq:Szekeres_I_eff_density_alt}, Lemma \ref{Lemma:Lemma_1}
would be violated.
\end{proof}
\begin{rem}
The condition $\text{sign}\left(F\left(r,p,q\right)\right)=\text{sign}\left(W\left(\tau_{0},r,p,q\right)\right)$
of Theorem \ref{prop:Proposition_1} is equivalent to consider that
the effective energy-momentum tensor \eqref{eq:Energy_momentum_tensor_dust}
verifies the weak energy condition at the initial time $\tau_{0}$. 
\end{rem}
\begin{rem}To clarify the statement of the Theorem \ref{prop:Proposition_1},
let us consider the simpler case of an effective dust fluid composed
of only one type of fermions. In this case, from Eqs.~\eqref{eq:eff_density_final_polytrope}
and \eqref{eq:Szekeres_I_eff_density_alt} we can solve for $\rho$,
 such that
\begin{equation}
\rho\left(r,p,q\right)=\frac{\bar{\rho}}{2}\left(1\pm\sqrt{1-\frac{4r^{2}F}{\bar{\rho}R^{2}W}}\right)\,,\label{eq:WR2}
\end{equation}
hence, given $\text{sign}\left(F\left(r,p,q\right)\right)=\text{sign}\left(W\left(\tau_{0},r,p,q\right)\right)$,
from the results of Appendix \ref{sec:Appendix-A}, if, for each $r$,
either $W\left(\tau,r,p,q\right)$ or $R\left(\tau,r\right)$ go to
zero for some instant $\tau\in\left]\tau_{0},\tau_{c}\left(r\right)\right]$,
the energy density profile for the fluid will take complex values,
hence considered unphysical.
\end{rem}

\subsection{Szekeres space-times: KS-like\label{subsec:Szekeres_class_II}}

The KS-like Szekeres models are characterized by the general line
element \cite{Szekeres1} 
\begin{equation}
ds^{2}=-d\tau^{2}+\left[X\left(\tau,r\right)-R\left(\tau\right)\frac{E_{1}\left(r,p,q\right)}{E_{0}\left(p,q\right)}\right]^{2}dr^{2}+\frac{R\left(\tau\right)^{2}}{E_{0}\left(p,q\right)^{2}}\left(dp^{2}+dq^{2}\right)\,,\label{eq:Szekeres_II_line_element}
\end{equation}
where the function $R\left(\tau\right)$ verifies 
\begin{equation}
\dot{R}\left(\tau\right)^{2}=\frac{2M}{R\left(\tau\right)}+k_{1}\,,\label{eq:Szekeres_II_Friedmann}
\end{equation}
with $M$ being an arbitrary constant to be given as initial
data and the constant $k_{1}$ is determined by 
\begin{equation}
\left(\partial_{p}E_{0}\right)^{2}+\left(\partial_{q}E_{0}\right)^{2}-E_{0}\left(\partial_{p}^{2}E_{0}+\partial_{q}^{2}E_{0}\right)=k_{1}\,.
\end{equation}
Moreover, the function $X\left(\tau,r\right)$ is a solution of following
equation: 
\begin{equation}
X\,\ddot{R}+\dot{R}\,\dot{X}+R\,\ddot{X}=\frac{k_{2}\left(r\right)}{2}\,,\label{eq:Szekeres_II_X_evolution_equation}
\end{equation}
with 
\begin{equation}
2\partial_{q}E_{0}\partial_{q}E_{1}+2\partial_{p}E_{0}\partial_{p}E_{1}-\left(\partial_{q}^{2}E_{0}+\partial_{p}^{2}E_{0}\right)E_{1}-E_{0}\left(\partial_{p}^{2}E_{1}+\partial_{q}^{2}E_{1}\right)=k_{2}(r)\,,
\end{equation}
where we have omitted some functional dependencies to avoid saturating
the notation.

From the EFE, the effective mass-energy density of the spinning cloud
verifies 
\begin{equation}
\rho_{\text{eff}}=\frac{G\left(r,p,q\right)}{R^{2}\left(\tau\right)H\left(\tau,r,p,q\right)}\,,\label{eq:Szekeres_II_energy_density}
\end{equation}
where 
\begin{equation}
H\left(\tau,r,p,q\right)=X\left(\tau,r\right)\,E_{0}\left(p,q\right)-R\left(\tau\right)\,E_{1}\left(r,p,q\right)\,.\label{eq:Szekeres_II_function_W}
\end{equation}
Such space-times are completely determined by the initial functions
$E_{0}\left(p,q\right)$and $E_{1}\left(r,p,q\right)$ and the constant
$M$.

\subsubsection{Regularity conditions and the influence of spin in singularity formation\label{subsec:Szekeres_class_II_results}}

As in the previous section, we shall consider some further constraints
on the initial regularity of the space-time:
\begin{assumption}
\label{assu:assumption_2}\leavevmode 
\begin{enumerate}
\item At the initial time, $\tau=\tau_{0}$, the space-time is non-singular. 
\item For every triplet $(r,p,q)$, we assume $H\left(\tau_{0},r,p,q\right)$
and $\rho_{\text{eff}}\left(\tau_{0},r,p,q\right)$ to be non-null. 
\end{enumerate}
\end{assumption}
Now, given the above initial regularity constraints and looking at
Eqs.~\eqref{eq:Szekeres_II_energy_density} and \eqref{eq:Szekeres_II_function_W},
as well as to the solutions in Appendix \ref{sec:Appendix-A},
it can be seen that, following the same reasoning applied in the proofs
of subsection \ref{subsec:Szekeres_class_I_results},
we have the following result: 
\begin{prop}
\label{prop:Proposition_2}Given a Szekeres space-time with line element
\eqref{eq:Szekeres_II_line_element} filled with an uncharged perfect
fluid composed only of fermionic particles, characterized by an equation of state
such that, the fluid effectively behaves as dust, if assumptions \ref{assu:assumption_2}
are verified and $\text{sign}\left(G\left(r,p,q\right)\right)=\text{sign}\left(H\left(\tau_{0},r,p,q\right)\right)$,
then a curvature singularity will not form.
\end{prop}

\section{Special Cases\label{sec:Special_Cases}}

In this section, we shall discuss the effects of spin in the formation
of singularities for some particular cases of the Szekeres space-times,
where the analysis of the previous section can be extended.

\subsection{Lemaître-Tolman-Bondi space-time}

The Lemaître-Tolman-Bondi (LTB) space-time is a solution of the EFE
for a spherically symmetric neutral dust source characterized, in
co-moving coordinates, by the line element 
\begin{equation}
ds^{2}=-d\tau^{2}+\frac{R'\left(\tau,r\right)^{2}}{1+f\left(r\right)}dr^{2}+R\left(\tau,r\right)^{2}d\Omega^{2}\,,\label{eq:LTB_line_element}
\end{equation}
where $d\Omega^{2}\equiv d\theta^{2}+\sin^{2}\theta\,d\varphi^{2}$
represents the line element of the unit 2-sphere. The function $R\left(\tau,r\right)$
verifies Eq.~\eqref{eq:Szekeres_I_Friedmann} and represents the
circumference radius at an instant $\tau$ and radial coordinate $r$
and $f\left(r\right)>-1$ is an arbitrary $C^{2}$-function of the
radial coordinate. The LTB metric can be found from the Szekeres solution,
Eq.~\eqref{eq:Szekeres_I_line_element}, by setting $E^{2}\left(r,p,q\right)=\left(1+p^{2}+q^{2}\right)/4$
and $\epsilon=1$.

From the EFE, we find that Eq.~\eqref{eq:Szekeres_I_eff_density_alt}
is re-written for the LTB space-time as 
\begin{align}
\rho_{\text{eff}}\left(\tau,r\right) & =\frac{r^{2}\,F\left(r\right)}{R^{2}R^{\prime}}\,,\label{eq:LTB_eff_density_alt}
\end{align}
where $F\left(r\right)$ is another arbitary function of $r$ and
it is related to the function $M\left(r\right)$, Eq.~\eqref{eq:Szekeres_I_Friedmann},
by $M'\left(r\right)=4\pi r^{2}\,F\left(r\right)$, or
\begin{equation}
M\left(r\right)=4\pi\int_{0}^{r}F\left(\mathrm{r}\right)\,\mathrm{r}^{2}\,d\mathrm{r}.\label{eq:LTB_M_integral_expression}
\end{equation} 
It is worth remarking that, in
this case, the function $M\left(r\right)$ represents the mass contained
within a shell with coordinate $r$ (see e.g. \cite{Landau}).

In the case of a LTB space-time, the function $W\left(\tau,r\right)\equiv R'\left(\tau,r\right)$
and in the coordinate system defined by Eq.~\eqref{eq:coord_system_specified}
we have $W\left(\tau_{0},r\right)=+1$. Therefore, for the LTB
space-time imposing $F\left(r\right)>0$ is equivalent to impose $M'\left(r\right)>0$.

The result in Theorem \ref{prop:Proposition_1}, albeit important, does not provide an answer
in the cases where $F\left(r\right)<0$. One might think that such
cases are unphysical since the effective energy conditions 
would be violated. There are, however,
known physical phenomena that seem to violate the energy conditions
(see e.g. \cite{Visser}). Moreover, cases of collapsing space-times
with negative mass have been studied several times in the past (see e.g. \cite{RMann,Bubble} and references therein).

Let us then consider that $M'\left(r\right)$ takes negative values for all $r$. In this case we have, from Eq.~\eqref{eq:LTB_M_integral_expression},
that the mass function $M\left(r\right)$ is negative for all $r$.
Assuming that $f\left(r\right)\geq-2M/r$, otherwise Eq.~\eqref{eq:Szekeres_I_Friedmann}
has no real solutions, it is clear, from Eq.~\eqref{eq:Szekeres_I_Friedmann}, that one of two scenarios will occur: either $\dot{R}\left(\tau_{0},r\right)>0$,
in which case the effective dust matter will continue to expand indefinitely;
or $\dot{R}\left(\tau_{0},r\right)<0$ where at some instant of time,
say $\tau_{b}$, $\dot{R}\left(\tau_{b},r\right)$ will go to zero.
From Eq.~\eqref{eq:Szekeres_I_Friedmann}, it is easy to show that $\ddot{R}(\tau,r)$
is always positive, therefore, at $\tau_{b}$, the function $R$ will
have a minimum. Hence, the system will bounce back and start to expand
indefinitely. 

We summarize these conclusions in the next proposition, which softens the conditions of Theorem \ref{prop:Proposition_1} for the case of LTB:
\begin{prop}
\label{prop:Proposition_3}Given a LTB space-time composed of a spherically symmetric, uncharged, collapsing perfect fluid, composed only of fermionic particles, characterized by an equation of state such
that, the fluid effectively behaves as dust, if Assumptions \ref{assu:assumption_1}
are verified and $\text{sign}\left(M'\left(r\right)\right)$ is the same for all $r$ within
the space-time, the circumferential radius function $R\left(\tau,r\right)$
will not go to zero.
\end{prop}
Aside the simpler case where $M'\left(r\right)<0$ for all $r$, there may be configurations where there are regions of the space-time with positive effective energy density and other regions with negative effective energy density, however, it is easy to show that such configurations are solutions of the Einstein field equations only if surface layers, separating the different regions, are present. Such cases shall not be considered here.

To conclude this subsection, we note that if $M'\left(r\right)$ takes negative
values, shell-crossing singularities may occur.

\subsubsection{The evolution of the collapse\label{subsec:LTB_evolution_of_the_collapse}}

In the previous subsection it was shown that, under our assumptions, 
and if the sign of the function $M'\left(r\right)$ is the same for all $r$,
an uncharged effective dust cloud composed only of fermionic particles,
in a LTB space-time will not form a shell-focusing singularity. Then,
for $\dot{R}(\tau_{0},r)<0$,  and assuming that no shell-crossing singularities occur for $\text{sign}\left(M'(r)\right)<0$, there are only three possibilities for
the behavior of the uncharged effective dust cloud: 
\begin{enumerate}[label=\roman*)]
\item there are no global in time solutions of the EFE for a collapsing
uncharged effective dust matter; 
\item the gravitational collapse will lead to a bounce of the matter cloud; 
\item the gravitational collapse of the matter cloud will settle in a stable
configuration. 
\end{enumerate}
As we shall see below, the third scenario will never occur, and the
first and second cases may occur depending on the initial data:

{\bf (I)} Consider
the initial data  $M\left(r\right)>0$ and $f\left(r\right)\geq0$ or $M\left(r\right)=0$ and $f\left(r\right)>0$. 
In both cases, from Eq.~\eqref{eq:Szekeres_I_Friedmann},
the function $\dot{R}$ can not go to zero, that is, for a given $r$,
$R\left(\tau,r\right)$ is either a strictly decreasing or a strictly
increasing function of $\tau$. However, if $R$ is a decreasing function
in the $\tau$ coordinate, it will eventually go to zero, violating
the result of Proposition \ref{prop:Proposition_3}. Hence, in these
cases, the only possible physical solution is when the matter 
expands indefinitely.

{\bf (II)} Another possibility is the case where $M\left(r\right)<0$ and $f\left(r\right)>0$,
with $2M\left(r\right)/R\left(\tau_{0},r\right)+f\left(r\right)>0$
for a given $r$ within the matter cloud. In this case, from Eq.~\eqref{eq:Szekeres_I_Friedmann},
there are two possible behaviors depending on the initial value of
$\dot{R}\left(\tau\right)$: either $\dot{R}\left(\tau_{0},r\right)>0$,
in which case the effective dust cloud will continue to expand indefinitely;
or $\dot{R}\left(\tau_{0},r\right)<0$, where at some instant of time,
$\tau_{b}$, $\dot{R}\left(\tau_{b},r\right)$ will go to zero. From
Eq.~\eqref{eq:Szekeres_I_Friedmann}, it is easy to show that $\ddot{R}(\tau,r)$
is always positive. Therefore, at $\tau_{b}$, the function $R$ will
attain a minimum, different from zero, concluding that the system
will bounce and start to expand from then on. Note that this solution
is characterized by $M\left(r\right)<0$, therefore the weak energy
condition is not verified and there is no inconsistency with the result in ref.~\cite{Hehl1}.

{\bf (III)} Yet another possibility is the case where $M\left(r\right)>0$ and
$f\left(r\right)<0$, with $2M\left(r\right)/R\left(\tau_{0},r\right)+f\left(r\right)>0$:
if initially $\dot{R}\left(\tau_{0},r\right)<0$, then a singularity
will eventually form; on the other hand, if $\dot{R}\left(\tau_{0},r\right)>0$,
$\dot{R}$ will go to zero at some $\tau\in]\tau_{0},+\infty[$, however,
from Eq.~\eqref{eq:Szekeres_I_Friedmann} we see that $\ddot{R}(\tau,r)$
is always negative in this case, hence, the system will cease to expand
further and start to collapse, eventually forming a singularity. In
both scenarios the result in Theorem \ref{prop:Proposition_1}
is violated, hence, such initial data does not correspond to a physical
solution. This is because $M\left(r\right)>0$ implies that at some
region $F\left(r\right)>0$, therefore, if a singularity occurs the
energy density of the constituents of the fluid would have to take
complex values.

As examples, in Figs.~\ref{fig:M_neg_f_pos} and \ref{fig:M_pos_f_neg}, we show
the behavior of $R$ as a function of $\tau$ for various fixed values
of the coordinate $r$, depending on the initial sign of $\dot{R}\left(\tau_{0},r\right)$,
found by numerically solving Eq.~\eqref{eq:Szekeres_I_Friedmann}
in particular cases when $M\left(r\right)<0$ and $f\left(r\right)>0$ and
$M\left(r\right)>0$ and $f\left(r\right)<0$, respectively.

\begin{figure}
\subfloat[\label{fig:2.a}]{\includegraphics[scale=0.6]{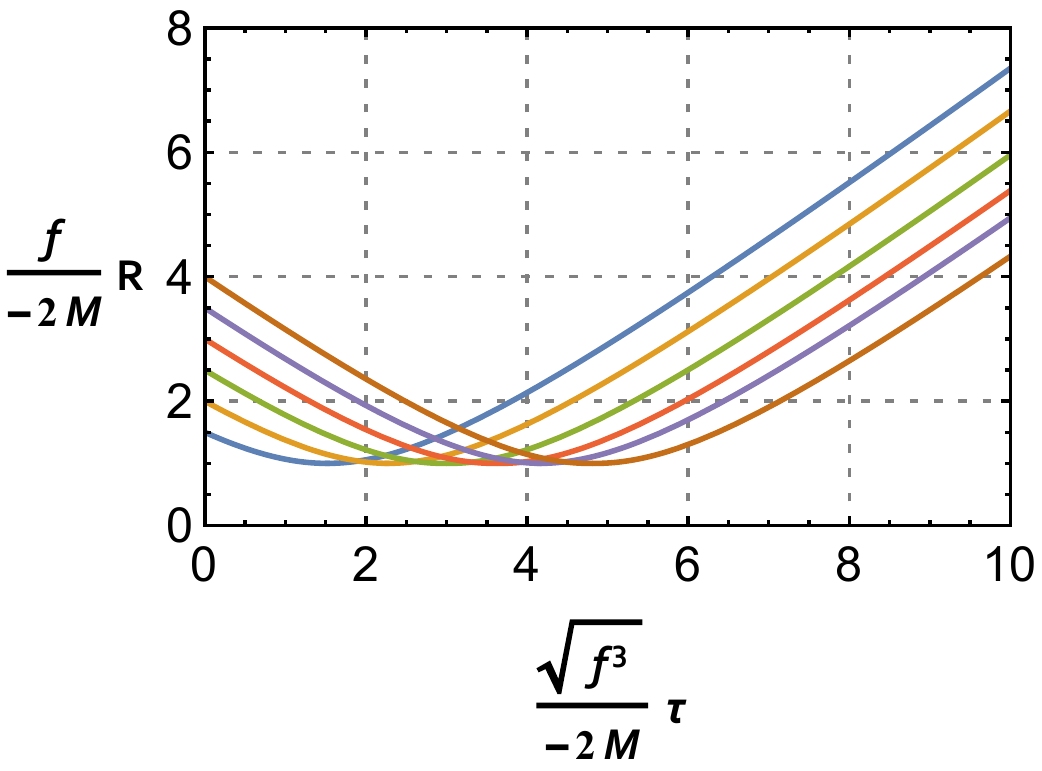}

}\hspace*{3cm}\subfloat[\label{fig:2.b}]{\includegraphics[scale=0.6]{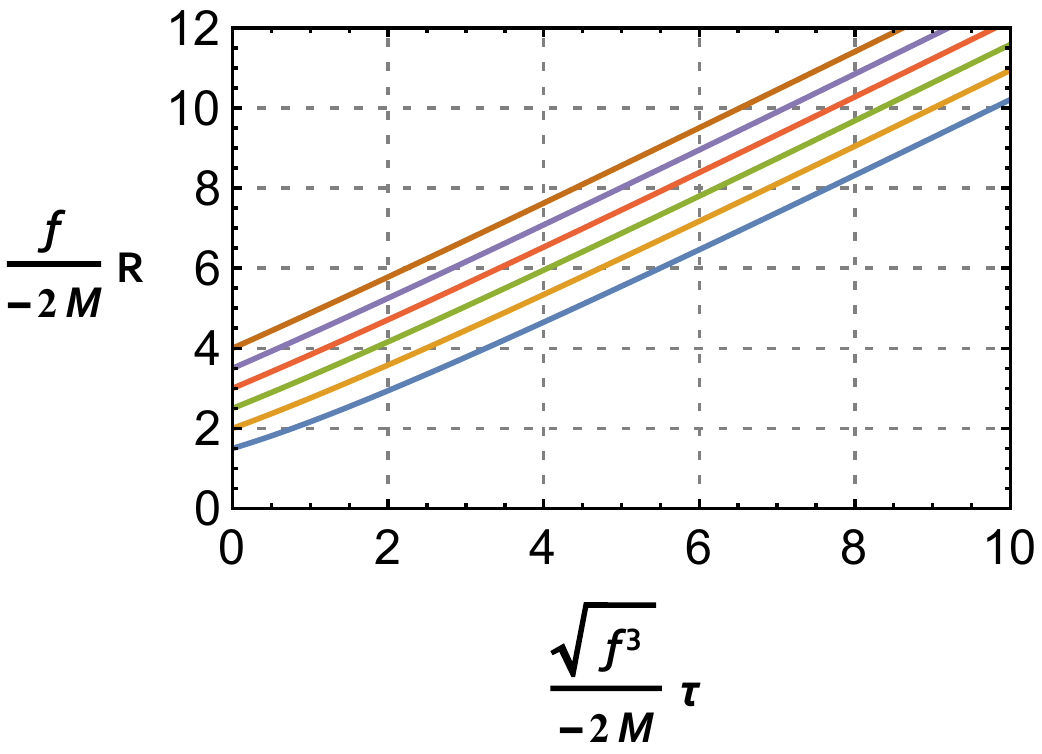}

}

\caption{\label{fig:M_neg_f_pos}Behavior of the function $R\left(\tau,r\right)$
for various fixed values of the coordinate $r$, depending on the
initial value of $\dot{R}$, in the case of $M\left(r\right)<0$ and
$f\left(r\right)>0$. Panel a) At the initial time $\dot{R}\left(\tau_{0},r\right)<0$.
Panel b) At the initial time $\dot{R}\left(\tau_{0},r\right)>0$.}
\end{figure}

\begin{figure}
\subfloat[\label{fig:3.a}]{\includegraphics[scale=0.6]{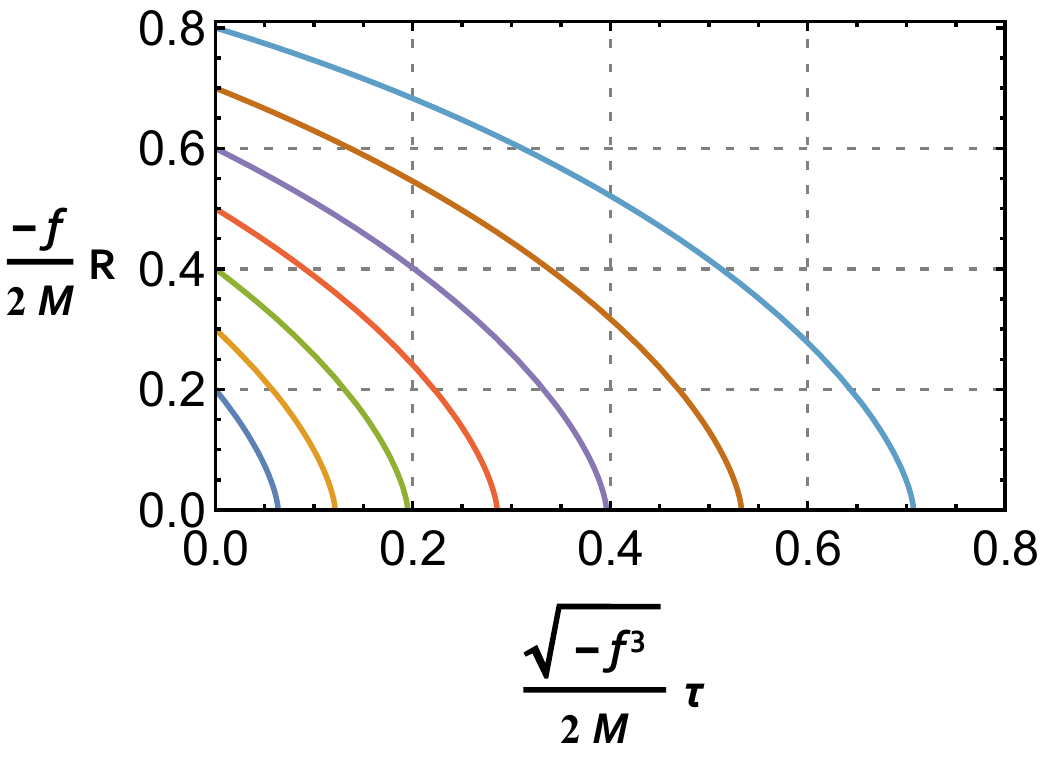}

}\hspace*{3cm}\subfloat[\label{fig:3.b}]{\includegraphics[scale=0.6]{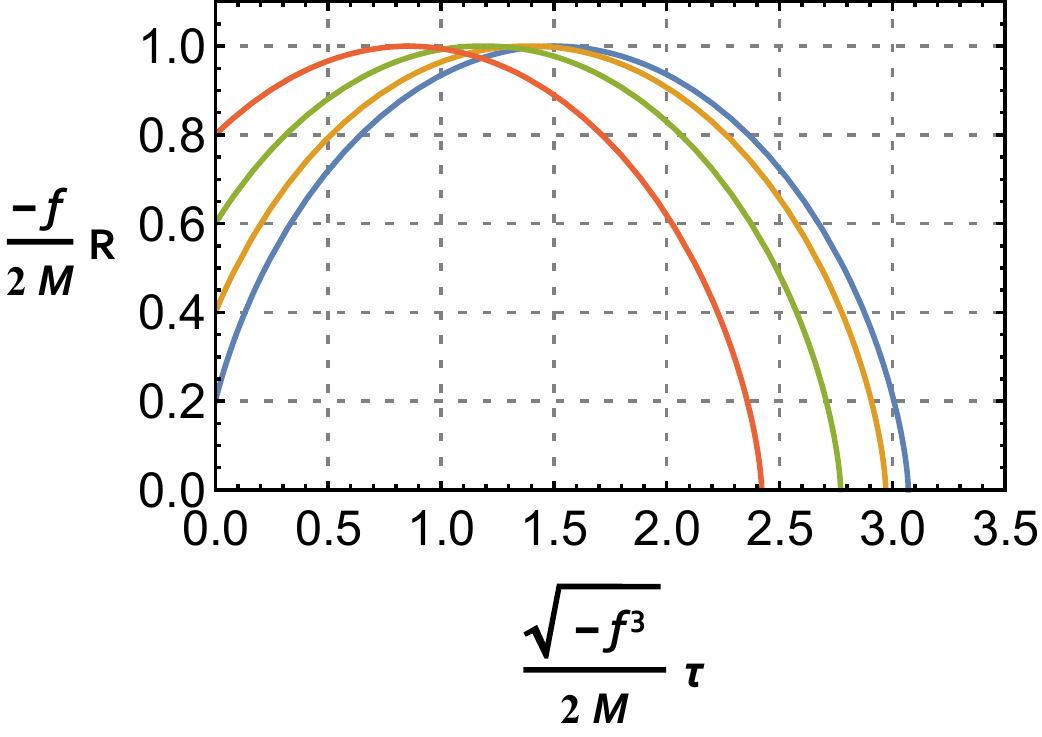}

}

\caption{\label{fig:M_pos_f_neg}Behavior of the function $R\left(\tau,r\right)$
for various values of the coordinate $r$ depending of the initial
value of $\dot{R}$ in the case of $M\left(r\right)>0$ and $f\left(r\right)<0$.
Panel a) At the initial time $\dot{R}\left(\tau_{0},r\right)<0$.
Panel b) At the initial time $\dot{R}\left(\tau_{0},r\right)>0$.}
\end{figure}

\subsection{Friedmann-Lemaître-Robertson-Walker space-time}

Another example of notable interest is the Friedmann-Lemaître-Robertson-Walker
(FLRW) model, where the space-time is parametrized by a spherically
symmetric spatially homogeneous and isotropic dust metric, such that
\begin{equation}
ds^{2}=-d\tau^{2}+a\left(\tau\right)^{2}\left(\frac{dr^{2}}{1+k\,r^{2}}+r^{2}d\Omega^{2}\right)\,,
\end{equation}
where $a\left(\tau\right)$ is the scale factor and $k=\left\{ -1,0,1\right\} $.
This solution is a particular case of the Szekeres line element, Eq.~\eqref{eq:Szekeres_I_line_element},
by setting $E^{2}\left(r,p,q\right)=\left(1+p^{2}+q^{2}\right)/4$,
$\epsilon=1$, $f\left(r\right)\equiv k\,r^{2}$ and $R\left(\tau,r\right)\equiv a\left(\tau\right)r$.
We then have for an effective dust cloud 
\begin{equation}
\rho_{\text{eff}}\left(\tau\right)=\frac{F_{0}}{a^{3}\left(\tau\right)}\,,
\end{equation}
with $F_{0}$ being a non-zero constant, and $a\left(\tau\right)$
satisfying the well-known Friedmann equation 
\begin{equation}
\dot{a}^{2}\left(\tau\right)=\frac{2M(r)}{r^{3}\,a\left(\tau\right)}+k\,,\label{scale-factor}
\end{equation}
where 
\begin{equation}
M(r)=M_{0}r^{3},\qquad\text{ with }~~~M_{0}=\frac{4\pi}{3}F_{0}=\frac{4\pi}{3}a\left(\tau_{0}\right)^{3}\rho_{\text{eff}}\left(\tau_{0}\right)\,.
\end{equation}
Moreover, this space-time is itself a particular case of the LTB space-time
studied in the previous subsection. Therefore, all the previous results
are valid for this particular case.

From Proposition \ref{prop:Proposition_3} we know that independently
of the sign of $M_{0}$, or equivalently $F_{0}$, a singularity does
not occur. As discussed in subsection \ref{subsec:LTB_evolution_of_the_collapse},
the actual evolution of space-time depends of the values of the constants
$k$ and $M_{0}$. More concretely, 
\begin{enumerate}
\item for $M_{0}>0\wedge k\geq0$, the space-time will expand indefinitely; 
\item for $M_{0}<0\wedge k>0$, depending on whether $\dot{a}\left(\tau_{0}\right)>0$
or $\dot{a}\left(\tau_{0}\right)<0$, the space-time will either expand
indefinitely or start by a collapsing phase which is followed by a
nonsingular bounce and then entering an expanding phase, respectively; 
\end{enumerate}
in all other cases, from Eq. \eqref{scale-factor}, we find that the
functions $a\left(\tau\right)$ or $\rho\left(\tau\right)$ will take
complex values, hence, the solutions to the Friedmann equations are
considered unphysical (see also Sec.6.2 of ref.~\cite{Griff_Podolsky}).

By comparison with our setup, we recall that in ref.~\cite{Trautman}, the effects of spin in singularity formation
was studied by considering a dust cloud composed of fermionic particles
in a FLRW space-time whose spins were assumed to be aligned in a given
preferred spatial direction. In that case, the field equations can
be explicitly solved and a closed form expression for the scale factor
$a\left(\tau\right)$ was found, showing explicitly that there is
a minimum positive value for $a\left(\tau\right)$ thus, concluding
that a singularity does not form. Moreover, cosmological as well as
astrophysical consequences of introducing spin and torsion in gravitation
has been studied in~\cite{cosmosptor} and~\cite{astrosptor,astrosptor1},
respectively. In the former, a closed homogeneous and isotropic universe
filled with fermionic matter has been considered and it was shown
that the effects of spin-torsion coupling produces a gravitational
repulsion in the early universe, preventing the formation of
cosmological singularity. In the latter papers, the effects of spin on the
dynamics of collapse for a closed~\cite{astrosptor} and
flat~\cite{astrosptor1} FLRW backgrounds has been studied. It was shown
that, under certain conditions, the formation of space-time singularities
can be avoided through a non-singular
bounce, which can either be hidden behind a horizon or visible to external
observers, depending on the initial radius or mass of the collapsing
body.

All the previous mentioned cases differ from our setup since we assume
that spins are randomly oriented and the fluid effectively behaves
as dust. However, as was shown above, our conclusions are similar.

\subsection{Senovilla-Vera space-time}

As a special inhomogeneous case of KS-like Szekeres space-times, we study the
Senovilla-Vera space-time \cite{Senovilla} characterized
by a line element of the form 
\begin{equation}
ds^{2}=-d\tau^{2}+dr^{2}+r^{2}d\varphi^{2}+\left(\beta+\frac{\tau^{2}+r^{2}}{\gamma}\right)^{2}dz^{2}\,,\label{eq:Senovilla_line_element}
\end{equation}
where $\gamma\in\mathbb{R}\backslash\left\{ 0\right\} $, $\beta\in\mathbb{R}_{\geq0}$,
$r>0$, $z\in\left]-\infty,+\infty\right[$ and $\varphi\in\left[0,2\pi\right[$. A space-time exterior to \eqref{eq:Senovilla_line_element} was found in \cite{I_Brito}.
Assuming that the space-time is permeated by effective dust, we find
that the effective mass-energy density is given by 
\begin{equation}
8\pi\rho_{\text{eff}}\left(\tau,r\right)=-\frac{4}{\beta\gamma+r^{2}+\tau^{2}}\,.\label{eq:Senovilla_efective_mass_density}
\end{equation}
In the cases where $\gamma>0\wedge\beta\geq0$ or $\gamma<0\wedge\gamma\beta+r^{2}>0$,
it is clear that there are no collapsing solutions. Let us then consider
the cases where $\gamma<0\wedge\gamma\beta+r^{2}<0$. Setting $R=1$,
$E_{0}=1$, $X=\tau^{2}/\gamma$ and $E_{1}=-\left(r^{2}+\beta\gamma\right)/\gamma$
in Eq.~\eqref{eq:Szekeres_II_line_element}, we recover the line element
\eqref{eq:Senovilla_line_element}. The functions $G$ and $H$, in
Eqs.~\eqref{eq:Szekeres_II_energy_density} and \eqref{eq:Szekeres_II_function_W},
are in this case 
\begin{equation}
H\left(\tau,r\right)=\beta+\frac{\tau^{2}+r^{2}}{\gamma}\,,\qquad\qquad G=-\frac{1}{2\pi\gamma}\,.
\end{equation}
Now, for $\gamma<0$ we see that the function $H\left(\tau,r\right)$
is a strictly decreasing function of the coordinate $\tau$. It is
then clear that, for $\gamma\beta+r^{2}<0$, such model is not a solution
of the EFE for effective dust, since this would imply the formation
of a curvature singularity when $\tau=\sqrt{\gamma\beta+r^{2}}$,
violating Proposition \ref{prop:Proposition_2}.

We can see this more clearly by considering
that the fluid is composed of only one
type of fermionic particles. In that case, from Eqs\@.~\eqref{eq:eff_density_final_polytrope}
and \eqref{eq:Senovilla_efective_mass_density} we have that the mass-energy
density of the perfect fluid is given by 
\begin{equation}
\rho\left(\tau,r\right)=\frac{\bar{\rho}}{2}\left[1\pm\sqrt{1+\frac{2}{\bar{\rho}\pi\left(\gamma\beta+\tau^{2}+r^{2}\right)}}\right]\,,
\end{equation}
where $\bar{\rho}_{i}$ is given by Eq.~\eqref{eq:Critical_density_general}
and the discriminant is assumed to be non-negative. From
this relation, we see that in the case of $\gamma<0\wedge\gamma\beta+r^{2}<0$,
when $\tau>\sqrt{\beta|\gamma|-r^{2}-\frac{2}{\bar{\rho}\pi}}$, the
mass-energy density of the fluid will take complex values, that is,
such solution is unphysical.

\subsection{LRS Bianchi type $I$ space-time}

As a final particular case, we will discuss the case of a locally
rotationally symmetric (LRS) Bianchi type I space-time with line element
\begin{equation}
ds^{2}=-d\tau^{2}+A\left(\tau\right)^{2}dx^{2}+B\left(\tau\right)^{2}\left(dy^{2}+dz^{2}\right),\label{eq:Bianchi_line_element}
\end{equation}
which is a spatially homogeneous particular case of the KS-like Szekeres metric.  
Interestingly, this space-time was used to study the influence
of spin in singularity formation by assuming that the space-time is
filled with a dust fluid (not effective dust) with non-null spin,
such that, the spins of the dust particles are aligned along a preferred
direction \cite{Stewart,Kop}.

In line with our previous examples we assume an effective dust fluid and, in that case, from the EFE, we find 
\begin{align}
A\left(\tau\right) & =\frac{c}{\left(b-\tau\right)^{\frac{1}{3}}}+\left(b-\tau\right)^{\frac{2}{3}}\,,\label{eq:Bianchi_A_t_final}\\
B\left(\tau\right) & =\left(b-\tau\right)^{\frac{2}{3}}\,,\label{eq:Bianchi_B_t_final}
\end{align}
with $b$ and $c$ being integration constants, and 
\begin{equation}
8\pi\,\rho_{\text{eff}}\left(\tau\right)=\frac{4}{3\Gamma\left(\tau\right)^{3}}\,,\label{eq:Bianchi_rho_eff_EFE}
\end{equation}
where 
\begin{equation}
\Gamma\left(\tau\right)^{3}=A\left(t\right)B\left(t\right)^{2}=\left(c+b-t\right)\left(b-t\right)\,.
\end{equation}
We note that the case $c=0$ corresponds to the flat FLRW metric in cylindrical
symmetry.

Now, Eqs.~\eqref{eq:Bianchi_line_element}, \eqref{eq:Bianchi_A_t_final}
and \eqref{eq:Bianchi_B_t_final} can be found from Eq.~\eqref{eq:Szekeres_II_line_element}
by setting $E_{0}=E_{1}=1$, $R\left(\tau\right)=B\left(\tau\right)$
and $X\left(\tau\right)=A\left(\tau\right)+B\left(\tau\right)$, hence
\begin{equation}
H\left(\tau\right)=A\left(\tau\right)=\frac{c}{\left(b-\tau\right)^{\frac{1}{3}}}+\left(b-\tau\right)^{\frac{2}{3}}\,,
\end{equation}
and the function $G=1/6\pi$. To apply the results found in the previous
section, at the initial time, for simplicity $\tau=0$, $H\left(0\right)$
must be positive (to match the sign of $G$) and finite, that is
\begin{equation}
\begin{cases}
c>-b & \text{, for \ensuremath{b>0}}\\
c<-b & \text{, for \ensuremath{b<0}}
\end{cases}\,.\label{eq:Bianchi_constraint_parameters}
\end{equation}
Therefore, if the constraints in Eq.~\eqref{eq:Bianchi_constraint_parameters}
are verified, Proposition \ref{prop:Proposition_2} tells us that
a shell-focusing or shell-crossing singularity will not be formed.
As in the previous subsection, this can be readily verified: In the
case $b<0$ this conclusion is trivial, since $A\left(\tau\right)$
and $B\left(\tau\right)$ will never be zero; On the other hand, for
\textbf{$b>0$}, assuming that the fluid is composed of only one type
of fermionic particles, the mass-energy density is given by 
\begin{equation}
\rho\left(\tau\right)=\frac{\bar{\rho}}{2}\left[1\pm\sqrt{1-\frac{2}{3\pi\bar{\rho}\,\Gamma\left(\tau\right)^{3}}}\right]\,,
\end{equation}
so that, as $\tau\to b$ or $\tau\to b+c$ (whichever occurs
first), the mass-energy density will take complex values, and
the resulting solution is unphysical.\textcolor{red}{{} }Let us also remark
that if Eq.~\eqref{eq:Bianchi_constraint_parameters} is not verified
but $\Gamma\left(0\right)\neq0$ then, 
\begin{equation}
\begin{cases}
\lim_{\tau\to b}\Gamma\left(\tau\right)\to-\infty & \text{, for \ensuremath{b>0\,\wedge\,c<-b} }\\
\lim_{\tau\to c+b}\Gamma\left(\tau\right)\to-\infty & \text{, for \ensuremath{b<0\,\wedge\,c>-b} }
\end{cases}\,,
\end{equation}
and all solutions of the EFE are real and a singularity will form.

\subsubsection{The Raychaudhuri equation}

To finish, let us use the LRS Bianchi type I space-time to discuss
a possible source of confusion that might arise when considering the
effects of spin in singularity formation. Consider a congruence of
geodesics in the LRS Bianchi type I space-time whose fiducial curve's
tangent vector is $v=\delta_{\tau}^{\alpha}\partial_{\alpha}$.
The shear scalar, $\sigma^{2}\equiv\sigma_{\alpha\beta}\sigma^{\alpha\beta}$,
where $\sigma_{\alpha\beta}\equiv\frac{1}{2}\left(\nabla_{\alpha}v_{\beta}+\nabla_{\beta}v_{\alpha}\right)$,
is given by 
\begin{equation}
\sigma^{2}=\frac{2}{3}\left(\frac{\dot{A}}{A}-\frac{\dot{B}}{B}\right)^{2}=\frac{2c^{2}}{3\Gamma^{6}}\,,\label{eq:Bianchi_shear_scalar}
\end{equation}
and the expansion scalar, $\theta\equiv\nabla_{\alpha}v^{\alpha}$,
\begin{equation}
\theta=\frac{\dot{A}}{A}+2\frac{\dot{B}}{B}=3\frac{\dot{\Gamma}}{\Gamma}\,.\label{eq:Bianchi_expansion_scalar}
\end{equation}
In this case, the vorticity $\omega_{\alpha\beta}\equiv\frac{1}{2}\left(\nabla_{\alpha}v_{\beta}-\nabla_{\beta}v_{\alpha}\right)$
is identically zero and the Raychaudhuri equation then reads 
\begin{equation}
\dot{\theta}+\frac{1}{3}\theta^{2}+\sigma^{2}+4\pi G\rho_{\text{eff}}=\dot{\theta}+3\left(\frac{\dot{\Gamma}}{\Gamma}\right)^{2}+\frac{2c^{2}}{3\Gamma^{6}}+\frac{2}{3\Gamma^{3}}=0\,.\label{eq:Bianchi_Raychaudhuri}
\end{equation}
Now, this result might seem to contradict Theorem \ref{prop:Proposition_1}
since, in a collapsing setting, it indicates that a singularity, in
the sense of a caustic, will always form. Let us, however, get back
to Eq.~\eqref{eq:Bianchi_rho_eff_EFE}. From Eqs.~\eqref{eq:Bianchi_shear_scalar}
and \eqref{eq:Bianchi_expansion_scalar} we find 
\begin{equation}
8\pi\,\rho_{\text{eff}}=\frac{\theta^{2}}{3}-\frac{\sigma^{2}}{2}\,.\label{eq:eq:Bianchi_rho_eff_alt}
\end{equation}
Substituting Eq.~\eqref{eq:eq:Bianchi_rho_eff_alt} in Eq.~\eqref{eq:Bianchi_Raychaudhuri},
\begin{equation}
\dot{\theta}+\frac{1}{2}\theta^{2}+\frac{3}{4}\sigma^{2}=0\,.\label{eq:Bianchi_ray_modified}
\end{equation}
Let us discuss this result: we see that relation \eqref{eq:Bianchi_ray_modified}
lost all information regarding the presence of spin, that is, this
equation is valid whether we are studying an effective spinning dust
in a Bianchi type $I$ using the Einstein-Cartan framework or a Bianchi
type $I$ space-time permeated by dust with no spin degrees of freedom
in GR. It is then clear that the Raychaudhuri equation alone may not
be enough to infer the possible formation of singularities, in the
sense that, in this case, although mathematically Eq.~\eqref{eq:Bianchi_ray_modified}
does imply the formation of a singularity (in the sense that $\Gamma\to0$),
it does not guarantee that the energy-density of the fluid, $\rho\left(\tau\right)$,
is a real function throughout the evolution of the space-time.

\section{Concluding remarks\label{sec:Conclusion}}

We have considered models of gravitational collapse of inhomogeneous and anisotropic 
(effective) dust fluid on a space-time described
by a Szekeres metric. We have found that, under certain conditions on the initial data,
the formation of a singularity may be avoided due to the presence
of spin. 
Comparing our results with those in the literature for spatially
homogeneous space-times, it was shown that not only the geometry of
the space-time, but also the equation of state of the fluid, play a pivotal role
in the evolution of the space-time and singularity formation. Moreover,
it was shown that even if the effective energy-momentum tensor of
the spinning fluid verifies the weak energy condition, a singularity
can be avoided. 

Some particular cases of the Szekeres model were considered
in order to either extend the previous results or show explicitly
the evolution of the effective dust space-time. The results found
for the various cases are summarized in Table \ref{tab:behavior_particular_solutions}.

Finally, we noted that evaluating the Raychaudhuri equation alone may
not be enough to infer to possible formation of a curvature singularity
in the Einstein-Cartan theory with spin torsion for physical scenarios,
since although a caustic may (mathematically) form, it is not guaranteed
that the quantities that describe the fluid will take real values
throughout the space-time evolution. Indeed, a crucial assumption in our analysis is the fact that the energy density remains real throughout
the evolution which, in our cases, also guarantees that there is no shell crossing.
 
The model of an effective dust fluid represents a critical case, providing
but a first step towards a deeper study of this important question:
In what conditions may spin effects prevent the formation of singularities?
We expect that deviations from the critical case will give rise to
a broader variety of dynamics and outcomes of gravitational collapse,
including the formation of black holes or naked singularities. In
the light of the results in this article, this appears to remain
an interesting open problem.

\begin{table}

\resizebox{17cm}{!}{%
\begin{tabular}{|c|c|>{\centering}m{4cm}|>{\raggedright}m{7cm}|>{\raggedright}m{4.5cm}|}
\hline 
\multicolumn{1}{|c|}{\rule{0pt}{12pt}\textbf{Particular solution}} & \hspace{8pt}\textbf{Parameters}\hspace{8pt} & \multicolumn{2}{c|}{\textbf{Possible behavior}} & \textbf{Other}\tabularnewline
\hline 
\multirow{4}{*}{\rule{0pt}{75pt}LTB} & \multirow{4}{*}{\rule{0pt}{75pt}$\left\{ M\left(r\right),\,f\left(r\right)\right\} $\hspace{3pt}} & \rule{0pt}{15pt}$M\left(r\right)>0\wedge\,f\left(r\right)\geq0$

\rule{0pt}{0pt} & \multirow{2}{7cm}{\rule{0pt}{22pt}The space-time will expand indefinitely.\rule{0pt}{0pt}} & \multirow{4}{4.5cm}{\rule{0pt}{75pt}If $M'\left(r\right)<0$, a shell-crossing singularity
may form.}\tabularnewline
\cline{3-3} 
 &  & \rule{0pt}{15pt}$M\left(r\right)=0\wedge\,f\left(r\right)>0$

\rule{0pt}{0pt} &  & \tabularnewline
\cline{3-4} 
 &  & \multirow{2}{4cm}{\rule{0pt}{22pt}\hspace{13pt}$M\left(r\right)<0\,\wedge\,f\left(r\right)>0$} & \rule{0pt}{20pt}If $\dot{R}\left(\tau_{0},r\right)\geq0$, the space-time will expand indefinitely.

\rule{0pt}{0pt} & \tabularnewline
 &  &  & \rule{0pt}{10pt}If $\dot{R}\left(\tau_{0},r\right)<0$, the space-time will collapse, bounce with $R\left(\tau,r\right)\neq0$ and
expand indefinitely.

\rule{0pt}{0pt} & \tabularnewline
\hline 
\multirow{2}{*}{\rule{0pt}{22pt}Senovilla-Vera} & \multirow{2}{*}{\rule{0pt}{20pt}$\left\{ \beta,\,\gamma\right\} $} & \rule{0pt}{15pt}$\gamma>0\wedge\beta\geq0$

\rule{0pt}{0pt} & \multirow{2}{7cm}{\rule{0pt}{22pt}The space-time will expand indefinitely.} & \tabularnewline
\cline{3-3} 
 &  & \rule{0pt}{15pt}$\gamma<0\wedge\gamma\beta+r^{2}>0$ 

\rule{0pt}{0pt} &  & \tabularnewline
\hline 
\multirow{3}{*}{\rule{0pt}{35pt}LRS Bianchi type I} & \multirow{3}{*}{\rule{0pt}{35pt}$\left\{ b,\,c\right\} $} & \rule{0pt}{15pt}$b<0\,\wedge\,c+b<0$

\rule{0pt}{0pt} & \rule{0pt}{15pt}The space-time will expand indefinitely.

\rule{0pt}{0pt} & \tabularnewline
\cline{3-5} 
 &  & \rule{0pt}{15pt}$b>0\,\wedge\,c+b<0$

\rule{0pt}{0pt} & \multirow{2}{7cm}{\rule{0pt}{22pt}A singularity will form in finite time.} & \multirow{2}{4.5cm}{}\tabularnewline
\cline{3-3} 
 &  & \rule{0pt}{15pt}$b<0\,\wedge\,c+b>0$

\rule{0pt}{0pt} &  & \tabularnewline
\hline 
\end{tabular}}

\caption{\label{tab:behavior_particular_solutions}Evolution of an effective
dust matter on various space-times, particular solutions of
the Szekeres metric. The space-times assumed to verify the premises
of Theorem \ref{prop:Proposition_1} or Proposition \ref{prop:Proposition_2}. All unmentioned possible initial data either lead to unphysical
solutions or correspond to the trivial static case.}
\end{table}

\section*{Acknowledgments}

AHZ is grateful to D. Malafarina and R. Goswami for useful discussions
and to F. W. Hehl for helpful correspondence. PL thanks IDPASC and
FCT-Portugal for financial support through Grant No. PD/BD/114074/2015.
FCM was supported by Portuguese Funds through FCT - Fundação para
a Ciência e Tecnologia, within the Projects UID/MAT/00013/2013 and
PTDC/MAT-ANA/1275/2014 as well as grant SFRH/BSAB/130242/2017. This work has
been supported financially by Research Institute for Astronomy \& Astrophysics of Maragha
(RIAAM) under research project No. 1/5750-61.

\appendix

\section{$\lim_{\tau\to\tau_{c}\left(r\right)}R{}^{2}R'=0$\label{sec:Appendix-A}}

In the proof of Lemma \ref{Lemma:Lemma_1} we have used the result
that, when $\tau_{c}\left(r\right)$ exists, then $\lim_{\tau\to\tau_{c}\left(r\right)}R{}^{2}R'=0$.
Since this result plays a key role in that proof, here we summarize
the solutions of the generalized Friedmann equation, Eq.~\eqref{eq:Szekeres_I_Friedmann}
(also Eq.~\eqref{eq:Szekeres_II_Friedmann}) and show how they can
be used in the proof.

The solutions to \eqref{eq:Szekeres_I_Friedmann} can be separated
in various sub-families depending on the combination of the signs
of the functions $M\left(r\right)$ and $f\left(r\right)$. It can
be easily seen that not all solutions are real (e.g. the case of $M\left(r\right),\,f\left(r\right)<0$).
The solutions of interest are then \cite{Szekeres1}

\begin{align}
R\left(\tau,r\right)=-2\frac{M\left(r\right)}{f\left(r\right)}\cos^{2}\left(\eta\right) & ,\,\eta+\frac{1}{2}\sin\left(2\eta\right)=\frac{\sqrt{-f\left(r\right)^{3}}}{2M\left(r\right)}\left(\tau-\tau_{0}\left(r\right)\right) & \text{, for }f\left(r\right)<0\text{ and }M\left(r\right)>0\label{Appendix_eq:Friedman_solution_1}\\
R\left(\tau,r\right)=-2\frac{M\left(r\right)}{f\left(r\right)}\cosh^{2}\left(\eta\right) & ,\,\eta+\frac{1}{2}\sinh\left(2\eta\right)=\frac{f\left(r\right)^{\frac{3}{2}}}{-2M\left(r\right)}\left(\tau-\tau_{0}\left(r\right)\right) & \text{, for }f\left(r\right)>0\text{ and }M\left(r\right)<0\label{Appendix_eq:Friedman_solution_2}\\
R\left(\tau,r\right)=2\frac{M\left(r\right)}{f\left(r\right)}\sinh^{2}\left(\eta\right) & ,\,\frac{1}{2}\sinh\left(2\eta\right)-\eta=\frac{f\left(r\right)^{\frac{3}{2}}}{2M\left(r\right)}\left(\tau-\tau_{c}\left(r\right)\right) & \text{, for }f\left(r\right)>0\text{ and }M\left(r\right)>0\label{Appendix_eq:Friedman_solution_3}\\
 & R\left(\tau,r\right)=\sqrt{f\left(r\right)}\left(\tau-\tau_{c}\left(r\right)\right) & \text{, for }f\left(r\right)>0\text{ and }M\left(r\right)=0\label{Appendix_eq:Friedman_solution_4}\\
 & R\left(\tau,r\right)=\left(\frac{9M\left(r\right)}{2}\right)^{\frac{1}{3}}\left(\tau-\tau_{c}\left(r\right)\right)^{\frac{2}{3}} & \text{, for }f\left(r\right)=0\text{ and }M\left(r\right)>0\label{Appendix_eq:Friedman_solution_5}\\
 & R\left(\tau,r\right)=\text{const.} & \text{, for }f\left(r\right)=0\text{ and }M\left(r\right)=0\label{Appendix_eq:Friedman_solution_6}
\end{align}
where some functional dependencies were omitted to simplify the notation.

Now, for the cases represented by Eqs.~\eqref{Appendix_eq:Friedman_solution_4}
- \eqref{Appendix_eq:Friedman_solution_5}, the computation of $\lim_{\tau\to\tau_{c}\left(r\right)}R{}^{2}R'$
can be done directly and verified to be zero.

Let us consider the case where $f\left(r\right)<0$ and $M\left(r\right)>0$,
Eq.~\eqref{Appendix_eq:Friedman_solution_1}. Taking the derivative
with respect to $r$ of the parametric equation for $\eta\left(\tau,r\right)$
we find 
\begin{equation}
\frac{\partial\eta}{\partial r}=\frac{1}{1+\cos\left(2\eta\right)}\frac{\partial}{\partial r}\left(\frac{\sqrt{-f\left(r\right)^{3}}}{2M\left(r\right)}\left(\tau-\tau_{0}\left(r\right)\right)\right)\,.\label{Appendix_eq:deta_dr}
\end{equation}
On the other hand, 
\begin{equation}
R'=-\frac{\partial}{\partial r}\left(\frac{2M\left(r\right)}{f\left(r\right)}\right)\cos^{2}\left(\eta\right)+4\frac{M\left(r\right)}{f\left(r\right)}\cos\left(\eta\right)\sin\left(\eta\right)\frac{\partial\eta}{\partial r}\,.\label{Appendix_eq:R_prime}
\end{equation}
Then, using Eqs.~\eqref{Appendix_eq:Friedman_solution_1}, \eqref{Appendix_eq:deta_dr}
and \eqref{Appendix_eq:R_prime}, we find 
\begin{equation}
\label{limit}
\begin{aligned}\lim_{\tau\to\tau_{c}\left(r\right)}R{}^{2}R' & =\lim_{\eta\to\frac{\pi}{2}}\left(\frac{2M\left(r\right)}{f\left(r\right)}\right)^{2}\cos^{4}\left(\eta\right)\left[4\frac{M\left(r\right)}{f\left(r\right)}\frac{\partial}{\partial r}\left(\frac{\sqrt{-f\left(r\right)^{3}}}{2M\left(r\right)}\left(\tau-\tau_{0}\left(r\right)\right)\right)\frac{\cos\left(\eta\right)\sin\left(\eta\right)}{1+\cos\left(2\eta\right)}-\frac{\partial}{\partial r}\left(\frac{2M\left(r\right)}{f\left(r\right)}\right)\cos^{2}\left(\eta\right)\right]\\
 & =-8\left(\frac{M\left(r\right)}{f\left(r\right)}\right)^{3}\lim_{\eta\to\frac{\pi}{2}}\,\cos^{3}\left(\eta\right)\sin\left(\eta\right)\,\frac{\partial}{\partial r}\left(\frac{\sqrt{-f\left(r\right)^{3}}}{2M\left(r\right)}\left(\tau-\tau_{0}\left(r\right)\right)\right)\\
 & =0\,,
\end{aligned}
\end{equation}
if the limit exists.

For the cases of $f\left(r\right)>0$ and $M\left(r\right)<0$, and
$f\left(r\right)=M\left(r\right)=0$ we see, from Eqs.~\eqref{Appendix_eq:Friedman_solution_2}
and \eqref{Appendix_eq:Friedman_solution_6}, that the function $R\left(\tau,r\right)$
is never zero, hence $\tau_{c}\left(r\right)$ does not exist.

In the case of $f\left(r\right)>0$ and $M\left(r\right)>0$, Eq.~\eqref{Appendix_eq:Friedman_solution_3},
the considered limit can be computed similarly to the case of $f\left(r\right)<0$
and $M\left(r\right)>0$ and, as in \eqref{limit}, the result is also zero.

\end{document}